\documentclass[twocolumn]{autart}    % Enable this line and disable the
\usepackage{graphicx,amssymb,amsmath,tikz} % Include this line if your
\newtheorem{theorem}{Theorem}[section]

\newtheorem{definition}{Definition}[section]
\newtheorem{example}{Example}[section]

\newtheorem{remark}{Remark}[section]
\newtheorem{proposition}{Proposition}[section]

\begin{document}

\begin{frontmatter}

\title{Strong current-state and initial-state opacity of \\ discrete-event systems \thanksref{footnoteinfo}} % Title, preferably not more

\thanks[footnoteinfo]{This work is supported in part by the National Natural Science Foundation of China under Grant 61903274, the Tianjin Natural Science Foundation of China under Grant 18JCQNJC74000, and the Alexander von Humboldt Foundation.}

\author[Han]{Xiaoguang Han}\ead{hxg-allen@163.com},    % Add the
\author[Zhang]{Kuize Zhang}\ead{kuize.zhang@campus.tu-berlin.de},               % e-mail address
\author[Han]{Jiahui Zhang}\ead{zjhtust@163.com},
\author[Li]{Zhiwu Li}\ead{zhwli@xidian.edu.cn},  % (ead) as shown
\author[Chen]{Zengqiang Chen}\ead{chenzq@nankai.edu.cn}

\address[Han]{College of Electronic Information and Automation, Tianjin University of Science and Technology, Tianjin 300222, China}
\address[Zhang]{Control Systems Group, Technical University of Berlin, Berlin 10587, Germany}
\address[Li]{Institute of Systems Engineering, Macau University of Science and Technology, Taipa 519020, China, and School of Electro-Mechanical Engineering, Xidian University, Xi'an 710071, China}
\address[Chen]{College of Artificial Intelligence, Nankai University, Tianjin 300350, China}

\begin{abstract}                          % Abstract of not more than 200 words.
Opacity, as an important property in information-flow security, characterizes the ability of a system to keep some secret information from an intruder.
In discrete-event systems, based on a standard setting in which an intruder has the complete knowledge of the system's structure, the standard versions of current-state opacity and initial-state opacity cannot perfectly characterize high-level privacy requirements.
To overcome such a limitation, in this paper we propose two stronger versions of opacity in partially-observed discrete-event systems, called \emph{strong current-state opacity} and \emph{strong initial-state opacity}.
Strong current-state opacity describes that an intruder never makes for sure whether a system is in a secret state at the current time, that is, if a system satisfies this property, then for each run of the system ended by a secret state, there exists a non-secret run whose observation is the same as that of the previous run.
Strong initial-state opacity captures that the visit of a secret state at the initial time cannot be inferred by an intruder at any instant.
Specifically, a system is said to be strongly initial-state opaque if for each run starting from a secret state, there exists a non-secret run of the system that has the same observation as the previous run has.
To verify these two properties, we propose two information structures using a novel concurrent-composition technique, which has exponential-time complexity $O(|X|^4|\Sigma_o||\Sigma_{uo}||\Sigma|2^{|X|})$,
where $|X|$ (resp., $|\Sigma|$, $|\Sigma_o|$, $|\Sigma_{uo}|$) is the number of states (resp., events, observable events, unobservable events) of a system.
\end{abstract}

                             % keyword list or with the
                                          % help of the Automatica
                                          % keyword wizard

\end{frontmatter}

\section{Introduction}\label{sec1}

Opacity is a confidentiality property, which characterizes the scenario whether some ``secret" of a system is ambiguous to a potential malicious intruder based on the evolution of the system.
In other words, an opaque system always holds the plausible deniability for its ``secret" during its execution.
Opacity can be adapted to capture a variety of privacy and security requirements in cyber-physical systems, including event-driven dynamic systems~\cite{Lafortune(2018),Hadjicostis(2020)} and time-driven dynamic systems~\cite{An(2020),Ramasubramanian(2020)}.

The notion of opacity initially appeared in the computer science literature~\cite{Mazare(2004)} to analyze cryptographic protocols.
Whereafter, various versions of opacity were characterized in the context of discrete-event systems (DESs).
For instance, in automata-based models, several notions of opacity have been proposed and studied, which include current-state opacity~\cite{Saboori(2007)}, initial-state opacity~\cite{Saboori(2013)}, $K$-step opacity~\cite{Saboori(2011a)}, infinite-step opacity~\cite{Saboori(2012)}, and language-based opacity~\cite{Lin(2011),Wu(2013)}.
Some more efficient algorithms to check them have also been provided in~\cite{Wu(2013),Yin(2017),Lan(2020)}.
Furthermore, when a given system is not opaque, its opacity enforcement problem has been extensively investigated using a variety of techniques, including supervisory control~\cite{Dubreil(2010),Saboori(2012b),Tong(2018)}, insertion or edit functions~\cite{Ji(2018),Ji(2019a),Ji(2019b),Yin(2020)}, dynamic observers~\cite{Zhang(2015)}, etc.
Recently, verification and/or enforcement of opacity have been extended to other classes of settings, including Petri nets~\cite{Bryans(2008),Tong(2017),Cong(2018)}, stochastic systems~\cite{Saboori(2014),Keroglou(2017),Yin(2019a)}, modular systems~\cite{Tong(2019)}, networked systems~\cite{Yin(2019b),Yang(2019)}, etc.
In recent literature~\cite{Zhang(2019),Yin(2021)}, the authors studied opacity preserving (bi)simulation relations using an abstraction-based technique for nondeterministic transition systems and metric systems, respectively.
Some applications of opacity in real-world systems have also been provided in the literature, see, e.g.,~\cite{Saboori(2011b),Wu(2014),Bourouis(2017),Lin(2020)}.

In automata-based models, different notions of opacity can capture different security requirements.
For example, current-state opacity (CSO) characterizes that an intruder cannot determine for sure whether a system is currently in a secret state.
In Location-Based Services (LBS), a user may want to hide his/her initial location or his/her location (e.g., visiting a hospital or bank) at some specific previous instant.
Such requirements can be characterized by initial-state opacity (ISO) and $K$/infinite-step opacity ($K$/Inf-SO)\footnote{For convenience, the notion originally named CSO (resp., ISO, $K$-SO, and Inf-SO) is categorized as standard CSO (resp., standard ISO, standard $K$-SO, and standard Inf-SO) in this paper.}.
However, these four standard versions of opacity have some limitations in practice.
Specifically, they cannot capture the situation that an intruder can never infer for sure whether the system has passed through a secret state based on his/her observation.
In other words, even though a system is ``opaque" in the standard sense, the intruder may necessarily determine that a secret state must has been passed through.
To this end, in~\cite{Falcone(2015)}, a strong version of opacity called \emph{strong $K$-step opacity} is proposed to capture that the visit of a secret state cannot be inferred within $K$ steps.
Also, an algorithm that has complexity $O(2^{|X|+|X|^2})$ is provided to check it using so-called $K$-delay trajectory estimators and R-verifiers, where $|X|$ is the number of states of a system.

Inspired by~\cite{Falcone(2015)}, in this paper we are interested in defining and verifying two other strong versions of opacity in partially-observed nondeterministic finite-state automata (NFAs),
called \emph{strong current-state opacity} (SCSO) and \emph{strong initial-state opacity} (SISO), respectively.
They mean that if a run passes through a secret state at the current time (resp., at the initial time), there exists another run that never passes through a secret state and has the same observation as the previous run has.
Obviously, they have higher-level confidentiality than the standard versions.
Note that, by the definition of strong $K$-step opacity proposed in~\cite{Falcone(2015)}, we know readily that strong $K$-step opacity reduces to standard CSO when $K=0$.
Thus, we conclude that SCSO implies the strong $0$-step opacity but actually these two strong versions of opacity are incomparable when $K\geq 1$.
Therefore, the proposed notions of SCSO and SISO are clearly different from all the existing ones in the literature.
Further, we develop a novel method to verify SCSO and SISO from a span-new perspective.
Specifically, the main contributions of this paper are summarized below.

\begin{itemize}
\item To overcome the inadequacy that the standard CSO and ISO cannot perfectly characterize high-level privacy requirements, we propose two stronger versions of opacity, called SCSO and SISO.
      They can characterize higher-level confidentiality than the standard versions.
\item For SCSO, we construct an information structure using a novel concurrent-composition technique.
      And then, we propose a verification criterion of SCSO based on \emph{leaking secret states} defined by us in the proposed structure.
      This verification approach costs time $O(|X|^4|\Sigma_o||\Sigma_{uo}||\Sigma|2^{|X|})$, where $|X|$ (resp., $|\Sigma|, |\Sigma_o|, |\Sigma_{uo}|$) is the number of states (resp., events, observable events, unobservable events) of a system.
\item Regarding the verification problem of SISO, using the novel concurrent-composition technique, we propose another information structure in which the information that whether a system has passed through a secret initial state can be characterized.
    This approach has also time complexity $O(|X|^4|\Sigma_o||\Sigma_{uo}||\Sigma|2^{|X|})$.
\end{itemize}

The rest of this paper is arranged as follows.
Section~\ref{sec2} provides preliminaries needed in this paper, including the system model, the formal definitions of standard CSO and ISO, and their limitations on confidentiality.
Section~\ref{sec3} formalizes the notion of SCSO, and proposes a novel information structure for its verification.
In Section~\ref{sec4}, the notion of SISO is proposed, and its verification approach is reported.
In Section~\ref{sec5}, we conclude this paper with a brief discussion on how to use the proposed concurrent-composition structure to verify \emph{strong infinite-step opacity}, which can be seen as a strong version of the standard Inf-SO.

\section{Preliminaries}\label{sec2}

\subsection{System model}\label{subsec2.1}

A DES of interest is modeled as a \emph{nondeterministic finite-state automaton} (NFA) $G=(X,\Sigma,\delta, X_0)$, where $X$ is the finite set of states, $\Sigma$ is the finite set of events, $X_0\subseteq X$ is the set of initial states, $\delta: X\times\Sigma\rightarrow 2^{X}$ is the transition function, which depicts the system dynamics: given states $x,y\in X$ and an event $\sigma\in \Sigma$, $y\in\delta(x,\sigma)$ implies that there exists a transition labeled by $\sigma$ from $x$ to $y$.
Note that for $x\in X$ and $\sigma\in\Sigma$, if there exists no $y\in X$ such that $y\in\delta(x,\sigma)$, then $\delta(x,\sigma)=\emptyset$.
We can extend the transition function to $\delta :X\times\Sigma^{\ast}\rightarrow 2^{X}$ in the usual manner, where $\Sigma^{\ast}$ denotes the \emph{Kleene closure} of $\Sigma$, consisting of all finite sequences composed of the events in $\Sigma$ (including the empty sequence $\epsilon$), see, e.g., \cite{Cassandras(2008)} for more details on DESs.
We use $\mathcal{L}(G,x)$ to denote the language generated by $G$ from state $x$, i.e., $\mathcal{L}(G,x)=\{s\in\Sigma^{\ast}: \delta(x,s)\neq\emptyset\}$.
Therefore, the language generated by $G$ is $\mathcal{L}(G)=\cup_{x_0\in X_0}\mathcal{L}(G,x_0)$.
For a sequence $s\in\mathcal{L}(G)$, we denote its prefix closure by $Pr(s)$, i.e., $Pr(s)=\{w\in\mathcal{L}(G): (\exists w^\prime\in\Sigma^{\ast})[ww^\prime=s]\}$.
Further, for a prefix $w\in Pr(s)$, we use the notation $s/w$ to denote the suffix of $s$ after its prefix $w$.
We without loss of generality assume that $G$ is accessible, i.e., all its states are reachable from $X_0$.

As usual, we assume that the intruder can only see partially the system's behavior.
To this end, $\Sigma$ is partitioned into the set $\Sigma_o$ of observable events and the set $\Sigma_{uo}$ of unobservable events, i.e., $\Sigma_o\cup\Sigma_{uo}=\Sigma$ and $\Sigma_o\cap\Sigma_{uo}=\emptyset$.
The natural projection $P:\Sigma^{\ast}\rightarrow \Sigma_o^{\ast}$ is defined recursively by
($i$) $P(\epsilon)=\epsilon$,
($ii$) $P(s\sigma)=P(s)\sigma,\mbox{ if } \sigma\in \Sigma_o$,
($iii$) $P(s\sigma)=P(s),\mbox{ if } \sigma\in \Sigma_{uo}$, where $s\in\Sigma^\ast$.
We extend the natural projection $P$ to $\mathcal{L}(G)$ by $P(\mathcal{L}(G))=\{P(s)\in\Sigma_{o}^{\ast}: s\in\mathcal{L}(G)\}$.
We denote the inverse projection by a mapping $P^{-1}:\Sigma_o^{\ast}\rightarrow 2^{\Sigma^{\ast}}$, i.e., for any $\alpha\in\Sigma_{o}^{\ast}$, we have $P^{-1}(\alpha)=\{s\in\Sigma^{\ast}:P(s)=\alpha\}$.

To study opacity of $G=(X,\Sigma,\delta,X_0)$, we assume that $G$ has a set of secret states, denoted by $X_{S}\subset X$.
We use $X_{NS}$ to denote the set of non-secret initial states.
Let $s=s_{1}s_{2}\ldots s_{n}\in\Sigma^{\ast}$ and $x_0\in X_0, x_i\in X$, $i=1,2,\ldots,n$.
If $x_{k+1}\in\delta(x_k,s_{k+1})$, $0\leq k\leq n-1$, we call $x_0\stackrel{s_1}{\rightarrow}x_1\stackrel{s_2}{\rightarrow}x_2\stackrel{s_3}{\rightarrow}\cdots\stackrel{s_n}{\rightarrow}x_n$ a \emph{run} generated by $G$ from $x_0$ to $x_n$ under $s$.
We write $x_0\stackrel{s}{\rightarrow}x_n$ (resp., $x_0\stackrel{s}{\rightarrow}$) when $x_1,x_2,\ldots,x_{n-1}$ (resp., $x_1,x_2,\ldots,x_n$) are not specified.
Note that there may be more than one run under a sequence $s\in\mathcal{L}(G,x_0)$ based on nondeterminism of $G$.
In other words, $x_0\stackrel{s}{\rightarrow}x_n$ may denote some runs.
$x_0\stackrel{s_1}{\rightarrow}x_1\stackrel{s_2}{\rightarrow}x_2\stackrel{s_3}{\rightarrow}\cdots\stackrel{s_n}{\rightarrow}x_n$ is called \emph{a non-secret run} if $x_i\in X\backslash X_S$, $i=0,1,2,\cdots,n$.

\subsection{Standard current-state opacity and initial-state opacity}\label{subsec2.2}

In this subsection, we first recall the formal definitions of standard CSO and ISO in~\cite{Saboori(2007)} and~\cite{Saboori(2013),Wu(2013)}, respectively.
And then, we discuss their limitations for characterizing information security.

\begin{definition}[Standard CSO~\cite{Saboori(2007)}]\label{de:2.1}
Given a system $G=(X,\Sigma,\delta,X_0)$, a projection map $P$ w.r.t. the set $\Sigma_o$ of observable events, and a set $X_{S}\subset X$ of secret states,
$G$ is said to be current-state opaque (CSO)\footnote{For brevity, the terminology ``CSO" (resp., ISO), is used as the acronym of both ``current-state opacity" (resp., ``initial-state opacity") and ``current-state opaque" (resp., ``initial-state opaque"), adapted to the context.} (w.r.t. $\Sigma_o$ and $X_{S}$), if for all $x_0\in X_0$ and for all $s\in\mathcal L(G,x_0)$ such that $\delta(x_0,s)\cap X_S\neq\emptyset$, it holds
\begin{equation}\label{eq:1}
\begin{split}
(\exists x^\prime_{0}\in X_0)(\exists t\in \mathcal L(G, x^\prime_0))[\delta(x^\prime_0,t)\cap (X\backslash X_S)\neq\emptyset \\
\wedge P(s)=P(t)].
\end{split}
\end{equation}
\end{definition}

Note that, in~\cite{Saboori(2007)} the notion of standard CSO was defined for deterministic finite-state automata.
Obviously, this property can be readily extended to the nondeterministic setting as in Definition~\ref{de:2.1}.

\begin{definition}[Standard ISO~\cite{Saboori(2013)}]\label{de:2.2}
Given a system $G=(X,\Sigma,\delta,X_0)$, a projection map $P$ w.r.t. the set $\Sigma_o$ of observable events, and a set $X_{S}\subset X$ of secret states,
$G$ is said to be initial-state opaque (ISO) (w.r.t. $\Sigma_o$ and $X_{S}$), if for all $x_0\in X_0\cap X_S$ and for all $s\in \mathcal{L}(G,x_0)$, it holds
\begin{equation}\label{eq:2}
(\exists x^\prime_0\in X_{NS})(\exists t\in\mathcal{L}(G, x^\prime_0)) [P(s)=P(t)].
\end{equation}
\end{definition}

\begin{remark}\label{re:2.1}
Many results on the verification and enforcement of CSO and ISO have been obtained so far.
However, the aforementioned two categories of opacity may not be able to characterize higher-level privacy requirements in some practical applications.
The following two examples illustrate their limitations.
\end{remark}

\begin{example}\label{ex:2.1}
Let us consider the system $G$ shown in Fig.~\ref{Fig1}, in which the set of secret states is $X_S=\{x_4,x_5\}$.
\begin{figure}[!ht]
  \centering
  \includegraphics[scale=0.75]{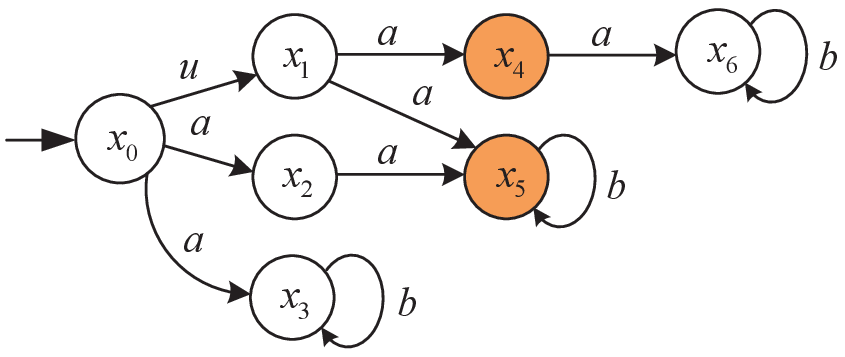}
  \caption{The system $G$ considered in Example~\ref{ex:2.1}, where $\Sigma_o=\{a,b\}$, $\Sigma_{uo}=\{u\}$ and $X_0=\{x_0\}$.}
  \label{Fig1}
\end{figure}
After observing the sequence $ab^k$, $k\geq 0$, $G$ can be in any state of $\{x_2, x_3, x_4, x_5\}$ (resp., $\{x_3, x_5\}$) when $k=0$ (reps., when $k\geq 1$).
For the sequences $s=aab^n$ and $t=uaab^n$, $n\geq 0$, we have $\delta(x_0,s)=\{x_5\}$, $\delta(x_0,t)=\{x_6\}$ and P(s)=P(t).
This means that, for each observation generated by $G$, the intruder can never infer precisely whether $G$ is currently at secret states $x_4$ and/or $x_5$.
By Definition~\ref{de:2.1}, $G$ is standard CSO w.r.t. $\Sigma_o$ and $X_S$.
Since an intruder has the perfect knowledge of the structure of $G$, he/she immediately concludes that $G$ must have visited secret state $x_4$ or $x_5$ after observing $aab^n$, $n\geq 0$.
This leads to some limitations that the security requirement characterized Definition~\ref{de:2.1} is not sufficiently strong to forbid the secret from being revealed.
\end{example}

\begin{example}\label{ex:2.2}
Consider the system $G$ shown in Fig.~\ref{Fig2}, where the set of secret states is $X_S=\{x_2\}$.
\begin{figure}[!ht]
  \centering
  \includegraphics[scale=0.75]{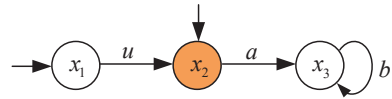}
  \caption{The system $G$ considered in Example~\ref{ex:2.2}, where $\Sigma_o=\{a,b\}$, $\Sigma_{uo}=\{u\}$ and $X_0=\{x_1,x_2\}$.}
  \label{Fig2}
\end{figure}
According to Definition~\ref{de:2.2}, we conclude readily that $G$ is standard ISO.
However, $G$ cannot keep the high-level secret based on the privacy condition characterized by Definition~\ref{de:2.2} since after observing the sequence $ab^n$ ($n\geq 0$) an intruder unambiguously determines that $G$ must have passed through secret initial state $x_2$.
\end{example}

The phenomena in Examples~\ref{ex:2.1} and~\ref{ex:2.2} indicate that the standard versions of CSO and ISO have some limitations to some extent for characterizing high-level information security of a system.
In order to overcome the aforementioned limitations, in this paper we propose two strong versions of opacity, called \emph{strong current-state opacity} (SCSO) and \emph{strong initial-state opacity} (SISO).
Note that these two stronger notions of opacity provide high-level confidentiality.

\section{Strong current-state opacity}\label{sec3}

In this section, we formally formulate the strong version of CSO for a system.
And then, we propose a novel approach to verify it.

\subsection{Notion of strong current-state opacity }\label{subsec3.1}

\begin{definition}[SCSO]\label{de:3.1}
Given a system $G=(X,\Sigma,\delta,$ $X_0)$, a projection map $P$ w.r.t. the set $\Sigma_o$ of observable events, and a set $X_{S}\subset X$ of secret states,
$G$ is said to be strongly current-state opaque (SCSO)\footnote{The terminology ``SCSO" is the acronym of both ``strong current-state opacity" and ``strongly current-state opaque", which depends on the context.} (w.r.t. $\Sigma_o$ and $X_{S}$), if for all $x_0\in X_0$ and for all $s\in\mathcal L(G,x_0)$ such that $\delta(x_0,s)\cap X_S\neq\emptyset$,
there exists a non-secret run $x^\prime_0\stackrel{t}{\rightarrow}x$ such that $P(t)=P(s)$, where $x^\prime_0\in X_0$, $x\in X$.
\end{definition}

Note that if $G$ is a deterministic finite-state automaton, SCSO in Definition~\ref{de:3.1} can be rephrased as follows:
$G$ is said to be SCSO (w.r.t. $\Sigma_o$ and $X_{S}$) if
\begin{equation*}
\begin{split}
& (\forall x_0\in X_0, \forall s\in\mathcal L(G,x_0): \delta(x_0,s)\in X_S)\\
& (\exists x^\prime_0\in X_0, \exists t\in\mathcal L(G,x^\prime_0))[(P(s)=P(t))\wedge\\
& (\forall \bar{t}\in Pr(t))[\delta(x^\prime_0,\bar{t})\notin X_S]].
\end{split}
\end{equation*}

Now we give a physical interpretation of Definition~\ref{de:3.1}.
The notion of SCSO captures that if a sequence $s\in\mathcal L(G)$ is generated and a secret state is reached,
then there exists a sequence $t\in\mathcal L(G)$ that has the same projection as $s$ and its execution generates a run that never passes through a secret state.
SCSO has a higher-level confidentiality than standard CSO.
In other words, SCSO implies standard CSO, but the converse implication does not hold.
For instance, the system $G$ in Example~\ref{ex:2.1} is standard CSO, but not SCSO since after observing the sequence $aa$ an intruder certainly concludes that $G$ has passed secret state $x_4$ or $x_5$.
This results in the following proposition for SCSO and CSO whose straightforward proof is omitted.

\begin{proposition}\label{pro:3.1}
If system $G$ is SCSO w.r.t. $\Sigma_o$ and $X_S$, then it is also standard CSO w.r.t. $\Sigma_o$ and $X_S$.
The converse is not true.
\end{proposition}

\begin{remark}\label{re:3.1}
In~\cite{Falcone(2015)}, the authors proposed the notion of strong $K$-step opacity for a deterministic finite-state automaton $G$.
Specifically, $G$ is said to be strongly $K$-step opaque (w.r.t. $\Sigma_o$ and $X_{S}$) if for all $x_0\in X_0$ and for all $s=s_1s_2\in\mathcal{L}(G,x_0)$ such that $\delta(x_0,s_1)\cap X_S\neq\emptyset$ and $|P(s_2)|\leq K$,
there exists $x^\prime_0\in X_0$ and $t\in\mathcal L(G,x^\prime_0)$ such that $P(t)=P(s)$ and for all $\bar{t}\in Pr(t)$, if $|P(t/\bar{t})|\leq K$, then $\delta(x^\prime_0,\bar{t})\in X\backslash X_S$.
Now we show that SCSO and strong $K$-step opacity are incomparable when $K\geq 1$.
Obviously, strong $K$-step opacity does not imply SCSO based on their definitions.
Also, the converse does not hold.
Let us consider the system $G$ shown in Fig.~\ref{Fig2}, we conclude readily that $G$ is SCSO, but not strongly $K$-step opaque for any $K\geq 1$.
\end{remark}

In the following Subsection~\ref{subsec3.2}, we study the verification of SCSO based on a novel verification approach, which has time complexity $O(|X|^4|\Sigma_o||\Sigma_{uo}||\Sigma|2^{|X|})$.

\subsection{Verifying strong current-state opacity}\label{subsec3.2}

In this subsection, we focus on the verification of SCSO in Definition~\ref{de:3.1}.
In order to obtain the main result, we need to introduce three notions of \emph{non-secret subautomaton}, \emph{observer} and \emph{concurrent composition} for a given system.

Given a system $G=(X,\Sigma,\delta,X_0)$ and a set $X_S\subset X$ of secret states,
we construct a non-secret subautomaton by deleting all its secret states from $G$, which is denoted by $G_{dss}$, where ``$dss$" stands for the acronym of ``\emph{deleting secret states}".
Note that when we delete a secret state $x_s\in X_S$, all transitions attached to $x_s$ are also deleted. Formally,
\begin{equation}\label{eq:3}
G_{dss}=(X_{dss},\Sigma_{dss},\delta_{dss},X_{dss,0}),
\end{equation}
where
\begin{itemize}
  \item $X_{dss}=\{x\in X\backslash X_S: \mbox{ there is a non-secret run } x_0\stackrel{s}{\rightarrow}x\mbox{ for some }x_0\in X_{NS} \mbox{ and } s\in\Sigma^\ast\}$ stands for the set of states;

  \item $\Sigma_{dss}=\{\sigma\in\Sigma: \exists x,x^{\prime}\in X_{dss} \mbox{ s.t. } x^{\prime}\in\delta(x,\sigma)\}$ stands for the set of events;
  \item $\delta_{dss}: X_{dss}\times\Sigma_{dss}\rightarrow 2^{X_{dss}}$ stands for the transition function, i.e., for all $x,x^{\prime}\in X_{dss} \mbox{ and } \sigma\in\Sigma_{dss}, x^{\prime}\in \delta_{dss}(x,\sigma) \mbox{ if } x^{\prime}\in \delta(x,\sigma)$;
  \item $X_{dss,0}=X_{NS}$ stands for the set of initial states.
\end{itemize}

Next, we recall the \emph{observer} of $G_{dss}$ that is defined by
\begin{equation}\label{eq:4}
Obs(G_{dss})=(X^{obs}_{dss},\Sigma^{obs}_{dss},\delta^{obs}_{dss},X^{obs}_{dss,0}),
\end{equation}
where $X^{obs}_{dss}\subseteq 2^{X_{dss}}\backslash \emptyset$ stands for the set of states;
$\Sigma^{obs}_{dss}$ stands for the observable event set $\Sigma_{dss,o}$ of $\Sigma_{dss}$;
$\delta^{obs}_{dss}: X^{obs}_{dss}\times\Sigma^{obs}_{dss}\rightarrow X^{obs}_{dss}$ stands for the (partial) deterministic transition function defined as follows: for any $q\in X^{obs}_{dss}$ and $\sigma\in\Sigma^{obs}_{dss}$,
we have $\delta^{obs}_{dss}(q,\sigma)=\{x^{\prime}\in X_{dss}: \exists x\in q, \exists w\in\Sigma_{dss,uo}^\ast \mbox{ s.t. }x^{\prime}\in\delta_{dss}(x,\sigma w)\}$ if it is nonempty, where $\Sigma_{dss,uo}=\Sigma_{dss}\backslash\Sigma_{dss,o}$;
$X^{obs}_{dss,0}=\{x\in X_{dss}: \exists x_0\in X_{dss,0}, \exists w\in\Sigma_{dss,uo}^\ast \mbox{ s.t. }x\in\delta_{dss}(x_0,w)\}$ stands for the (unique) initial state.
For brevity, we only consider the accessible part of observer $Obs(G_{dss})$.

Note that the widely-used \emph{observer} $Obs(G_{dss})$ is an important information structure that can track all possible states of $G_{dss}$ consistent with the current observation.
Specifically, $Obs(G_{dss})$ captures the observable behavior of $G_{dss}$, that is, $\mathcal{L}(Obs(G_{dss}))=P(\mathcal{L}(G_{dss}))$.
The time complexity of computing observer $Obs(G_{dss})$ for $G_{dss}$ is $O(|X|^2|\Sigma_o||\Sigma_{uo}|2^{|X|})$.
We refer the reader to \cite{Cassandras(2008),Shu(2007)} for more details on the notion of observer.

\begin{example}\label{ex:3.1}
Consider the system $G$ depicted in Fig.~\ref{Fig3}, where the set of secret states is $X_S=\{x_2,x_4\}$.
The non-secret subautomaton $G_{dss}$ and the observer $Obs(G_{dss})$ for $G$ are shown in Fig.~\ref{Fig4}.
\begin{figure}[!ht]
  \centering
  \includegraphics[scale=0.75]{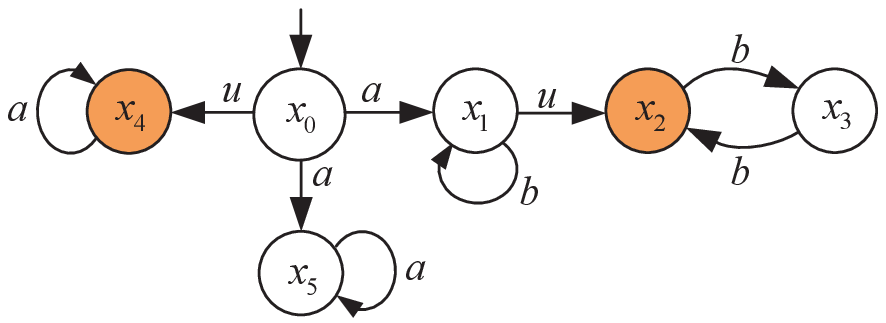}
  \caption{The system $G$ considered in Example~\ref{ex:3.1}, where $\Sigma_o=\{a,b\}$, $\Sigma_{uo}=\{u\}$ and $X_0=\{x_0\}$.}
  \label{Fig3}
\end{figure}

\begin{figure}[!ht]
  \centering
  \includegraphics[scale=0.75]{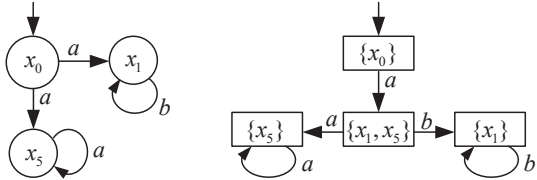}
  \caption{The non-secret subautomaton $G_{dss}$ (left) of $G$ shown in Fig.~\ref{Fig3} and the observer $Obs(G_{dss})$ (right) of $G_{dss}$.}
  \label{Fig4}
\end{figure}
\end{example}

Now we are ready to introduce a novel information structure called the \emph{concurrent composition} of system $G$ and observer $Obs(G_{dss})$, which will be used to verify SCSO in Definition~\ref{de:3.1}.

\begin{definition}[Concurrent Composition]\label{de:3.2}
Given a system $G=(X,\Sigma,\delta,X_0)$ and a set $X_{S}\subset X$ of secret states, the concurrent composition of $G$ and $Obs(G_{dss})$ is an NFA
\begin{equation}\label{eq:5}
Cc(G,Obs(G_{dss}))=(X_{cc},\Sigma_{cc},\delta_{cc},X_{cc,0}),
\end{equation}
where
\begin{itemize}
  \item $X_{cc}\subseteq X\times 2^{X}$ stands for the set of states;
  \item $\Sigma_{cc}=\{(\sigma,\sigma): \sigma\in\Sigma_o\}\cup\{(\sigma,\epsilon): \sigma\in\Sigma_{uo}\}$ stands for the set of events;
  \item $\delta_{cc}: X_{cc}\times\Sigma_{cc}\rightarrow 2^{X_{cc}}$ is the transition function defined as follows: for any state $(x,q)\in X_{cc}$ and for any event $\sigma\in\Sigma$,
  \begin{itemize}
  \item [(i)] when $q\neq\emptyset$,\\
    (a) if $\sigma\in\Sigma_o$, then
   \begin{equation*}
   \begin{split}
  & \delta_{cc}((x,q),(\sigma,\sigma))=\{(x^\prime,q^\prime)\in X_{cc}: x^\prime\in\delta(x,\sigma),\\
  & q^\prime=\delta^{obs}_{dss}(q,\sigma) \mbox{ if }\delta^{obs}_{dss}(q,\sigma) \mbox{ is well-defined },q^\prime=\emptyset \\
  & \mbox{otherwise}\};
   \end{split}
   \end{equation*}
    (b) if $\sigma\in\Sigma_{uo}$, then
   \begin{equation*}
   \delta_{cc}((x,q),(\sigma,\epsilon))=\{(x^\prime,q)\in X_{cc}: x^\prime\in\delta(x,\sigma)\}.
  \end{equation*}
  \item [(ii)] When $q=\emptyset$,\\
     (a) if $\sigma\in\Sigma_o$, then
   \begin{equation*}
   \delta_{cc}((x,\emptyset),(\sigma,\sigma))=\{(x^\prime,\emptyset)\in X_{cc}: x^\prime\in\delta(x,\sigma)\};
   \end{equation*}
     (b) if $\sigma\in\Sigma_{uo}$, then
   \begin{equation*}
   \delta_{cc}((x,\emptyset),(\sigma,\epsilon))=\{(x^\prime,\emptyset)\in X_{cc}: x^\prime\in\delta(x,\sigma)\}.
  \end{equation*}
  \end{itemize}
  \item $X_{cc,0}=X_0\times\{X^{obs}_{dss,0}\}$ stands for the set of initial states.
\end{itemize}
\end{definition}

For a sequence $e\in\mathcal{L}(Cc(G,Obs(G_{dss})))$, we utilize the notations $e(L)$ and $e(R)$ to denote its left and right components, respectively.
A similar notation is employed to the language $\mathcal{L}(Cc(G,Obs(G_{dss})))$.
Further, we use $P(e)$ to denote $P(e(L))$ or $P(e(R))$ since $P(e(L))=P(e(R))$ for any $e\in\mathcal{L}(Cc(G,Obs(G_{dss})))$, which depends on the context.
Intuitively, $Cc(G,Obs(G_{dss}))$ characterizes the following two properties:
(i) $\mathcal{L}(Cc(G,$ $Obs(G_{dss})))(L)=\mathcal{L}(G)$; and (ii) it tracks a sequence $s\in\mathcal L(G)$ from $X_0$ and an observation $\alpha\in\mathcal L(Obs(G_{dss}))$ from $X^{obs}_{dss,0}$ (if exists) such that $P(s)=\alpha$.

\begin{remark}\label{re:3.2}
The idea of concurrent composition was used in~\cite{Zhang(2020)} to verify the so-called eventual strong detectability for a finite-state automaton.
However, there are differences between the construction proposed in~\cite{Zhang(2020)} and the above construction.
First, the construction proposed in~\cite{Zhang(2020)} is based on a system $G$, while the concurrent composition defined in~\eqref{eq:5} is based on two different settings.
Second, the two definitions of transition functions meet different requirements.
Specifically, in~\cite{Zhang(2020)}, the concurrent composition aggregates each pair of transition sequences of $G$ generating the same projection.
However, by Definition~\ref{de:3.2}, $Cc(G,Obs(G_{dss}))$ captures that for any sequence $s\in\mathcal L(G)$, whether or not there exists an observation $\alpha\in\mathcal L(Obs(G_{dss})))$ such that $P(s)=\alpha$.
\end{remark}

\begin{example}\label{ex:3.2}
Consider again the system $G$ depicted in Fig.~\ref{Fig3}, the concurrent composition $Cc(G,Obs(G_{dss}))$ for $G$ is shown in Fig.~\ref{Fig5}, where $Obs(G_{dss})$ is shown in Fig.~\ref{Fig4}.
For brevity, for each state in $Cc(G,Obs(G_{dss}))$, its first and second components are abbreviated by using their respective subscript sets.
For example, state $(4,\{1,5\})$ stands for state $(x_4,\{x_1,x_5\})$.
\begin{figure}[!ht]
  \centering
  \includegraphics[scale=0.75]{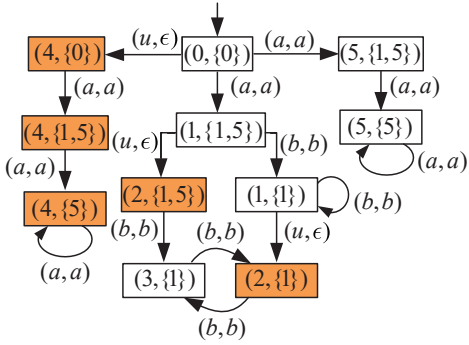}
  \caption{The concurrent composition $Cc(G,Obs(G_{dss}))$ of $G$ shown in Fig.~\ref{Fig3}, where $G_{dss}$ and $Obs(G_{dss})$ are shown in Fig.~\ref{Fig4}.}
  \label{Fig5}
\end{figure}
\end{example}

In order to present the first main result of this paper, we classify the states whose first components belong to $X_S$ in $Cc(G,Obs(G_{dss}))$ defined in~\eqref{eq:5} into two types:
\begin{itemize}
  \item \textbf{leaking secret states}: $(x,q)\in X_{cc}$ with $x\in X_S\wedge q=\emptyset$;
  \item \textbf{non-leaking secret states}: $(x,q)\in X_{cc}$ with $x\in X_S\wedge q\neq\emptyset$;
\end{itemize}

The leaking and non-leaking secret states are interpreted physically as follows.
If $(x,q)\in X_{cc}$ is a leaking secret state, then, by Definition~\ref{de:3.2}, in $G$ there exists $x_0\in X_0$ and $s\in\mathcal{L}(G,x_0)$ such that $x\in\delta(x_0,s)$, and there exists no observation $\alpha\in\mathcal L(Obs(G_{dss}))$ such that $P(s)=\alpha$.
This means that an intruder certainly concludes that $G$ has reached/passed through at least one secret state by observing $P(s)$.
If $(x,q)\in X_{cc}$ is a non-leaking secret state, then there exist $x_0, x^\prime_0\in X_0$, $s\in\mathcal{L}(G,x_0)$, and $t\in\mathcal{L}(G,x^\prime_0)$ such that:
(i) $x\in\delta(x_0,s)$; (ii) $s$ and $t$ have the same projection, i.e., $P(s)=P(t)$;
(iii) there exists a non-secret run $x^{\prime}_0\stackrel{t}{\rightarrow}x^{\prime}$ generated by $G$ from $x^\prime_0$ to $x^\prime$ under $t$.

Based on the above preparation, now we are ready to introduce the main result on verification of SCSO for a system.
It reveals that a system is SCSO if and only if its concurrent-composition structure does not contain any leaking secret state.

\begin{theorem}\label{th:3.1}
Given a system $G=(X,\Sigma,\delta,X_0)$, a projection map $P$ w.r.t. the set $\Sigma_o$ of observable events, and a set $X_{S}\subset X$ of secret states,
let $Cc(G,Obs(G_{dss}))$ be the concurrent composition of $G$ and $Obs(G_{dss})$.
$G$ is SCSO w.r.t. $\Sigma_o$ and $X_S$ if and only if there exists no leaking secret state in $Cc(G,Obs(G_{dss}))$.
\end{theorem}

\begin{proof}
$(\Rightarrow)$ By contrapositive, assume that there exists a leaking secret state $(x,\emptyset)$ in $Cc(G,Obs(G_{dss}))$.
Let $e\in\mathcal{L}(Cc(G,Obs(G_{dss})))$ be a sequence that leads to $(x,\emptyset)$ from $X_{cc,0}$.
Then, there exists an initial state $(x_0,X^{obs}_{dss,0})\in X_{cc,0}$ such that $(x,\emptyset)\in\delta_{cc}((x_0,X^{obs}_{dss,0}),e)$.
Further, by the construction of $Cc(G,Obs(G_{dss}))$, we have that $x\in\delta(x_0,e(L))$ and $\delta^{obs}_{dss}(X^{obs}_{dss,0},e(R))$ is not well-defined.
The former means $e(L)\in\mathcal{L}(G,x_0)$ and $\delta(x_0,e(L))\cap X_S\neq\emptyset$.
The latter means that for all $x^\prime_0\in X_{NS}$ and for all $t\in\mathcal{L}(G_{dss},x^\prime_0)$ with $P(t)=e(R)=P(e(L))$, it holds $\delta_{dss}(x^\prime_0,t)=\emptyset$.
By the construction of $G_{dss}$, we conclude that there exists no non-secret run in $G$ starting from $X_0$ with observation $P(e)$.
Therefore, by Definition~\ref{de:3.1}, $G$ is not SCSO w.r.t. $\Sigma_o$ and $X_S$.

$(\Leftarrow)$ Also by contrapositive, assume that $G$ is not SCSO w.r.t. $\Sigma_o$ and $X_S$.
By Definition~\ref{de:3.1}, there exists a run $x_0\stackrel{s}{\rightarrow}x_s$, where $x_0\in X_0$ and $x_s\in X_S$, and there exist no $x^\prime_0\in X_{NS}$ and $t\in\mathcal L(G,x^\prime_0)$ such that $P(t)=P(s)$ and $G$ generates a non-secret run starting from $x^\prime_0$ under $t$.
The former means $x_s\in\delta(x_0,s)$.
By the construction of $G_{dss}$, the latter means $\delta_{dss}(x^{\prime\prime}_0,t)=\emptyset$ for all $x^{\prime\prime}_0\in X_{NS}$ and for all $t\in\mathcal L(G,x^{\prime\prime}_0)$ with $P(t)=P(s)$.
This further means that $\delta^{obs}_{dss}(X_{dss,0}^{obs},P(t))$ is not well-defined.
By the construction of $Cc(G,Obs(G_{dss}))$, we conclude that there is $e\in\mathcal{L}(Cc(G,Obs(G_{dss})))$ with $e(L)=s$ and $e(R)=P(t)$ such that $(x_s,\emptyset)\in\delta_{cc}((x_0,X^{obs}_{dss,0}),e)$.
Therefore, $(x_s,\emptyset)$ a leaking secret state in $Cc(G,Obs(G_{dss}))$.
\end{proof}

\begin{example}\label{ex:3.3}
(1) For the system $G$ depicted in Fig.~\ref{Fig3}, there exists no leaking state in $Cc(G,Obs(G_{dss}))$ shown in Fig.~\ref{Fig5}.
By Theorem~\ref{th:3.1}, $G$ is SCSO w.r.t. $X_S$.
(2) For the system $G$ depicted in Fig.~\ref{Fig1}, its concurrent composition $Cc(G,Obs(G_{dss}))$ is shown in Fig.~\ref{Fig6} in which there exists a leaking secret states $(5,\emptyset)$.
By Theorem~\ref{th:3.1}, $G$ is not SCSO w.r.t. $\Sigma_o$ and $X_S$.
\begin{figure}[!ht]
  \centering
  \includegraphics[scale=0.75]{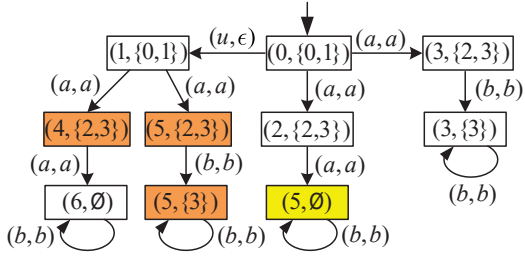}
  \caption{The concurrent composition $Cc(G,Obs(G_{dss}))$ of $G$ shown in Fig.~\ref{Fig1}, where $G_{dss}$ and $Obs(G_{dss})$ are shown in Fig.~\ref{Fig7}.}
  \label{Fig6}
\end{figure}
\begin{figure}[!ht]
  \centering
  \includegraphics[scale=0.75]{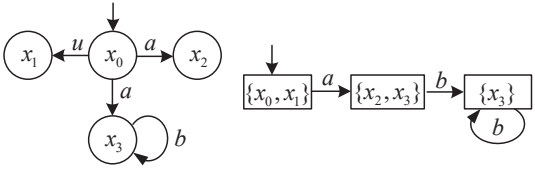}
  \caption{The non-secret subautomaton $G_{dss}$ (left) of $G$ in Fig.~\ref{Fig1} and the observer $Obs(G_{dss})$ (right) of $G_{dss}$.}
  \label{Fig7}
\end{figure}
\end{example}

We end this section by analyzing the time complexity of using Theorem~\ref{th:3.1} to verify SCSO.
Since $G$ has the complexity of $O(|X|^2|\Sigma|)$ and computing $Obs(G_{dss})$ costs time $O(|X|^2|\Sigma_o||\Sigma_{uo}|2^{|X|})$, the time complexity of constructing $Cc(G,Obs(G_{dss}))$ is
$O(|X|^4|\Sigma_o||\Sigma_{uo}||\Sigma|2^{|X|})$.
Consequently, the (worst-case) time complexity of using Theorem~\ref{th:3.1} to verify SCSO is $O(|X|^4|\Sigma_o||\Sigma_{uo}||\Sigma|2^{|X|})$.

\section{Strong initial-state opacity}\label{sec4}

In this section, we focus on the formulation of strong ISO and its verification problem for a system.

\subsection{Notion of strong initial-state opacity }\label{subsec4.1}

\begin{definition}[SISO]\label{de:4.1}
Given a system $G=(X,\Sigma,\delta,$ $X_0)$, a projection map $P$ w.r.t. the set $\Sigma_o$ of observable events, and a set $X_{S}\subset X$ of secret states,
$G$ is said to be strongly initial-state opaque (SISO)\footnote{Similar to the terminology ``SCSO", in this paper ``SISO" is the acronym of both ``strong initial-state opacity" and ``strongly initial-state opaque", which also depends on the context.} (w.r.t. $\Sigma_o$ and $X_{S}$), if for all $x_0\in X_0\cap X_S$ and for all $s\in\mathcal{L}(G,x_0)$,
there exists a non-secret run $x^\prime_0\stackrel{t}{\rightarrow}$ such that $P(t)=P(s)$, where $x^\prime_0\in X_{NS}$.
\end{definition}

Intuitively, SISO in Definition~\ref{de:4.1} captures that for each sequence $s\in\mathcal L(G)$ starting from any $x_0\in X_0\cap X_S$, there exists a sequence $t\in\mathcal L(G)$ starting from some $x^\prime_0\in X_{NS}$ such that $P(t)=P(s)$ and $G$ generates a non-secret run starting from $x^\prime_0$ under $t$.
Note that, SISO, compared with standard ISO in Definition~\ref{de:2.2}, has a higher-level confidentiality obviously.
In other words, SISO implies standard ISO, but the converse is not true.
Taking the system $G$ shown in Fig.~\ref{Fig2} for example, by Definitions~\ref{de:2.2} and~\ref{de:4.1}, we conclude readily that $G$ is standard ISO, but not SISO.

Next, we illustrate these two strong versions of opacity in Definitions~\ref{de:3.1} and~\ref{de:4.1} are incomparable.
In other words, SCSO does not imply SISO, and vice versa.
For instance, by Definitions~\ref{de:3.1} and~\ref{de:4.1}, the system $G$ depicted in Fig.~\ref{Fig2} is SCSO, but not SISO.
On the other hand, let us consider the system $G$ shown in Fig.~\ref{Fig8}.
We know that $G$ is SISO, but not SCSO.
\begin{figure}[!ht]
  \centering
  \includegraphics[scale=0.75]{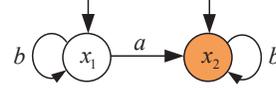}
  \caption{A system $G$ with $\Sigma=\Sigma_o=\{a,b\}$, $X_0=\{x_1,x_2\}$ and $X_{S}=\{x_2\}$.}
  \label{Fig8}
\end{figure}

\subsection{Verifying strong initial-state opacity }\label{subsec4.2}

Inspired by the verification of SCSO in Section~\ref{sec3}, we, in this subsection, use a minor variant of the concurrent composition approach to verify SISO of a system.

Consider a system $G=(X,\Sigma,\delta,X_0)$ in which the set of secret states is denoted by $X_S \subset X$.
$G_{dss}$ stands for the non-secret subautomaton obtained by deleting all secret states in $G$, which is defined in Eq.~\eqref{eq:4}.
$Obs(G_{dss})$ defined in Eq.~\eqref{eq:5} stands for the observer of $G_{dss}$.
To check SISO in Definition~\ref{de:4.1}, we construct an \emph{initial-secret subautomaton} from $G$, denoted by $\hat{G}$, as follows:

\begin{equation}\label{eq:6}
\hat{G}=(\hat{X},\hat{\Sigma},\hat{\delta},\hat{X}_0),
\end{equation}
where
\begin{itemize}
  \item $\hat{X}=\{x\in X: \exists x_0\in {X_0}\cap {X_S},\exists s\in\Sigma^\ast\mbox{ s.t. }x\in\delta(x_0,s)\}$ stands for the set of states;
  \item $\hat{\Sigma}=\{\sigma\in\Sigma:\exists x,x^{\prime}\in\hat{X} \mbox{ s.t. } x^{\prime}\in\delta(x,\sigma)\}$ stands for the set of events;
  \item $\hat{\delta}: \hat{X}\times\hat{\Sigma}\rightarrow 2^{\hat{X}}$ stands for the transition function, i.e., for all $x,x^{\prime}\in\hat{X}$ and $\sigma\in\hat{\Sigma}$, $x^{\prime}\in\hat{\delta}(x,\sigma)$ if $x^{\prime}\in\delta(x,\sigma)$;
  \item $\hat{X}_0={X_0}\cap {X_S}$ stands for the set of secret initial states.
\end{itemize}

Note that when we ``delete" a state that is inaccessible from the set $\hat{X}_0$ of secret initial states, we also delete all transitions attached to that state.
Obviously, $\hat{G}$ is a subautomaton of $G$ and is obtained by deleting all its states unreachable from $\hat{X}_0$.

Next, we construct the \emph{concurrent composition} of $\hat{G}$ and $Obs(G_{dss})$ for $G$, which also is an NFA, as follows.
\begin{equation}\label{eq:7}
Cc(\hat{G},Obs(G_{dss}))=(\hat{X}_{cc},\hat{\Sigma}_{cc},\hat{\delta}_{cc},\hat{X}_{cc,0}),
\end{equation}
where
\begin{itemize}
  \item $\hat{X}_{cc}\subseteq \hat{X}\times 2^{X}$ stands for the set of states;
  \item $\hat{\Sigma}_{cc}=\{(\sigma,\sigma):\sigma\in\hat{\Sigma}_o\}\cup\{(\sigma,\epsilon):\sigma\in\hat{\Sigma}_{uo}\}$ stands for the set of events, where $\hat{\Sigma}_o$ (resp., $\hat{\Sigma}_{uo}$) is the set of observable (resp., unobservable) events of $\hat{\Sigma}$;
  \item $\hat{\delta}_{cc}: \hat{X}_{cc}\times\hat{\Sigma}_{cc}\rightarrow 2^{\hat{X}_{cc}}$ is the transition function defined as follows: for any state $(x,q)\in \hat{X}_{cc}$, and for any event $\sigma\in\hat{\Sigma}$,
    \begin{itemize}
  \item [(i)] when $q\neq\emptyset$,\\
    (a) if $\sigma\in\hat{\Sigma}_o$, then
   \begin{equation*}
   \begin{split}
   & \hat{\delta}_{cc}((x,q),(\sigma,\sigma))=\{(x^\prime,q^\prime)\in \hat{X}_{cc}: x^\prime\in\hat{\delta}(x,\sigma),\\
   & q^\prime=\delta^{obs}_{dss}(q,\sigma) \mbox{ if } \delta^{obs}_{dss}(q,\sigma) \mbox{ is well-defined, } q^\prime=\emptyset \\
   & \mbox{otherwise}\};
   \end{split}
   \end{equation*}
    (b) if $\sigma\in\hat{\Sigma}_{uo}$, then
   \begin{equation*}
   \hat{\delta}_{cc}((x,q),(\sigma,\epsilon))=\{(x^\prime,q)\in \hat{X}_{cc}: x^\prime\in\hat{\delta}(x,\sigma)\}.
  \end{equation*}
  \item [(ii)] when $q=\emptyset$,\\
     (a) if $\sigma\in\hat{\Sigma}_o$, then
   \begin{equation*}
   \hat{\delta}_{cc}((x,\emptyset),(\sigma,\sigma))=\{(x^\prime,\emptyset)\in \hat{X}_{cc}: x^\prime\in\hat{\delta}(x,\sigma)\};
   \end{equation*}
     (b) if $\sigma\in\hat{\Sigma}_{uo}$, then
   \begin{equation*}
   \hat{\delta}_{cc}((x,\emptyset),(\sigma,\epsilon))=\{(x^\prime,\emptyset)\in\hat{X}_{cc}: x^\prime\in\hat{\delta}(x,\sigma)\}.
  \end{equation*}
  \end{itemize}
  \item $\hat{X}_{cc,0}=\hat{X}_0\times\{X^{obs}_{dss,0}\}$ stands for the set of initial states.
\end{itemize}

In order to present the verification criterion of SISO in Definition~\ref{de:4.1}, similarly to the argument in Section~\ref{sec3}, we classify the states in $Cc(\hat{G},Obs(G_{dss}))$ defined in \eqref{eq:7} into two types:

\begin{itemize}
  \item \textbf{leaking states}: for all $(x,q)\in\hat{X}_{cc}$ with $x\in\hat{X}\wedge q=\emptyset$;
  \item \textbf{non-leaking states}: for all $(x,q)\in\hat{X}_{cc}$ with $x\in\hat{X}\wedge q\neq\emptyset$;
\end{itemize}

We now give an interpretation of leaking and non-leaking states as follows:
if $(x,q)\in\hat{X}_{cc}$ is a leaking state, by the construction of $Cc(\hat{G},Obs(G_{dss}))$, then in $G$ there exists a secret initial state $x_0\in X_0\cap X_S$ and a sequence $s\in\mathcal{L}(G,x_0)$ such that $x\in\delta(x_0,s)$, and there exists no $x^{\prime}_0\in X_{NS}$ and $t\in\mathcal{L}(G,x^{\prime}_0)$ such that $x^{\prime}_0\stackrel{t}{\rightarrow}$ is a non-secret run and $P(t)=P(s)$.
If $(x,q)\in\hat{X}_{cc}$ is a non-leaking state, then in $G$ there exists $x_0\in X_0\cap X_S$, $x^\prime_0\in X_{NS}$, $s\in\mathcal{L}(G,x_0)$, and $t\in\mathcal{L}(G,x^\prime_0)$ such that: (i) $P(t)=P(s)$, and (ii) $G$ generates a non-secret run starting from $x^\prime_0$ under $t$.

\begin{theorem}\label{th:4.1}
Given a system $G=(X,\Sigma,\delta,X_0)$, a projection map $P$ w.r.t. the set $\Sigma_o$ of observable events, and a set $X_{S}\subset X$ of secret states,
let $Cc(\hat{G},Obs(G_{dss}))$ be the concurrent composition of $\hat{G}$ and $Obs(G_{dss})$.
$G$ is SISO w.r.t. $\Sigma_o$ and $X_S$ if and only if there exists no leaking state in $Cc(\hat{G},Obs(G_{dss}))$.
\end{theorem}

\begin{proof}
This proof is analogous to that of Theorem~\ref{th:3.1} by contraposition.

$(\Rightarrow)$ Assume that there exists a leaking state $(x,\emptyset)$ in $Cc(\hat{G},Obs(G_{dss}))$.
Let $e\in\mathcal{L}(Cc(\hat{G},Obs(G_{dss})))$ be a sequence that leads to $(x,\emptyset)$ from $\hat{X}_{cc,0}$.
Then, there exists an initial state $(x_0,X^{obs}_{dss,0})\in\hat{X}_{cc,0}$ such that $(x,\emptyset)\in\hat{\delta}_{cc}((x_0,X^{obs}_{dss,0}),e)$.
By the construction of $Cc(\hat{G},Obs(G_{dss}))$, one has $x_0\in{X_0}\cap{X_S}$, $x\in\hat{\delta}(x_0,e(L))$ and $\delta^{obs}_{dss}(X^{obs}_{dss,0},e(R))$ is not well-defined.
Since $\hat{G}$ is the initial-secret subautomaton of $G$, $x\in\hat{\delta}(x_0,e(L))$ means $x\in\delta(x_0,e(L))$.
However, $\delta^{obs}_{dss}(X^{obs}_{dss,0},e(R))$ not well-defined implies that for all $x^\prime_0\in X_{NS}$ and for all $t\in\mathcal{L}(G_{dss},x^\prime_0)$ with $P(t)=e(R)=P(e(L))$, it holds $\delta_{dss}(x^\prime_0,t)=\emptyset$.
By the construction of $G_{dss}$, we conclude that there do not exist a state $x^\prime_0\in X_{NS}$ and a sequence $t\in\mathcal{L}(G,x^\prime_0)$ with $P(t)=P(e)$ such that $G$ generates a non-secret run starting from $x^\prime_0$ under $t$.
Therefore, by Definition~\ref{de:4.1}, $G$ is not SISO w.r.t. $\Sigma_o$ and $X_S$.

$(\Leftarrow)$ Assume that $G$ is not SISO w.r.t. $\Sigma_o$ and $X_S$.
By Definition~\ref{de:4.1}, we know that: (1) in $G$ there exists a secret initial state $x_0\in X_0\cap X_S$ and a sequence $s\in\mathcal{L}(G,x_0)$,
and (2) there does not exist a sequence $t\in\mathcal L(G,x^\prime_0)$ for any $x^\prime_0\in X_{NS}$ such that $x^\prime_0\stackrel{t}{\rightarrow}$ is a non-secret run and $P(t)=P(s)$.
By the construction of $\hat{G}$, the former means that there exists a state $x\in\hat{X}$ such that $x\in\hat{\delta}(x_0,s)$.
By the construction of $G_{dss}$, the latter means $\delta_{dss}(x^{\prime\prime}_0,t)=\emptyset$ for all $x^{\prime\prime}_0\in X_{NS}$ and for all $t\in\mathcal L(\hat{G},x^{\prime\prime}_0)$ with $P(t)=P(s)$.
This further implies that $\delta^{obs}_{dss}(X_{dss,0}^{obs},P(t))$ is not well-defined.
By the construction of $Cc(\hat{G},Obs(G_{dss}))$, we conclude that there exists a sequence $e\in\mathcal{L}(Cc(\hat{G},Obs(G_{dss})))$ with $e(L)=s$ and $e(R)=P(t)$ such that $(x,\emptyset)\in\hat{\delta}_{cc}((x_0,X_{dss,0}^{obs}),e)$.
Therefore, $(x,\emptyset)$ a leaking state in $Cc(\hat{G},Obs(G_{dss}))$.
\end{proof}

\begin{remark}\label{re:4.1}
We discuss the time complexity of checking SISO for $G$ using Theorem~\ref{th:4.1}.
From Subsection~\ref{subsec3.2}, $Obs(G_{dss})$ has time complexity $O(|X|^2|\Sigma_o||\Sigma_{uo}|2^{|X|})$.
By Eq.~\eqref{eq:6}, the time complexity of constructing $\hat{G}$ from $G$ is $O(|X|^2|\Sigma|)$.
Therefore, the time complexity of constructing $Cc(\hat{G},Obs(G_{dss}))$ is $O(|X|^4|\Sigma_o||\Sigma_{uo}||\Sigma|2^{|X|})$.
This indicates that the time complexity of verifying SISO in Definition~\ref{de:4.1} is $O(|X|^4|\Sigma_o||\Sigma_{uo}||\Sigma|2^{|X|})$.
\end{remark}

\begin{example}\label{ex:4.1}
Let us consider the system $G$ depicted in Fig.~\ref{Fig9}, where the set of secret states is $X_S=\{x_1\}$.
The constructed automata $\hat{G}$, $G_{dss}$ and $Obs(G_{dss})$ for $G$ are shown in Fig.~\ref{Fig10}.
Further, the concurrent composition $Cc(\hat{G},Obs(G_{dss}))$ is shown in Fig.~\ref{Fig11} in which there exist two leaking states $(4,\emptyset)$ and $(5,\emptyset)$\footnote{Note that, similar to $Cc(G,Obs(G_{dss}))$, for each state in $Cc(\hat{G},Obs(G_{dss}))$, its first and second components are also abbreviated by using their respective subscript set. For instance, $(4,\emptyset)$ is the abbreviation of state $(x_4,\emptyset)$.}.
By Theorem~\ref{th:4.1}, $G$ is not SISO w.r.t. $\Sigma_o$ and $X_S$.
\begin{figure}[!ht]
  \centering
  \includegraphics[scale=0.75]{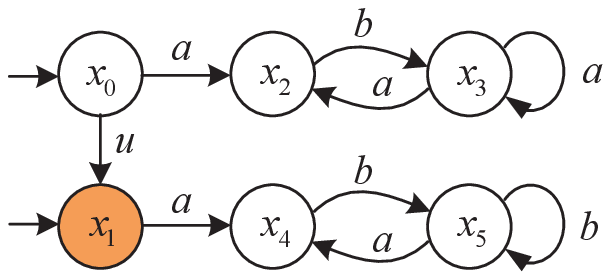}
  \caption{The system $G$ considered in Example~\ref{ex:4.1}, where $\Sigma_o=\{a,b\}$, $\Sigma_{uo}=\{u\}$ and $X_0=\{x_0,x_1\}$.}
  \label{Fig9}
\end{figure}
\begin{figure}[!ht]
  \centering
  \includegraphics[scale=0.75]{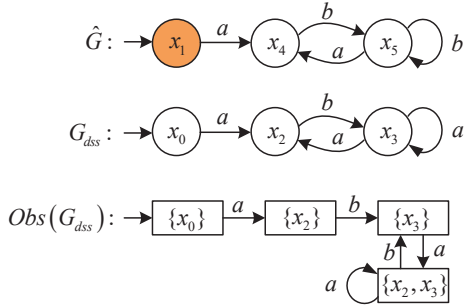}
  \caption{The constructed automata: $\hat{G}$ (above), $G_{dss}$ (middle) and $Obs(G_{dss})$ (below) for the system $G$ in Fig.~\ref{Fig9}.}
  \label{Fig10}
\end{figure}
\begin{figure}[!ht]
  \centering
  \includegraphics[scale=0.75]{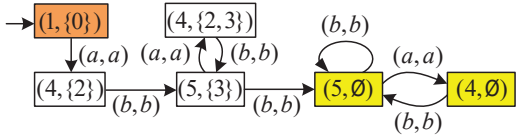}
  \caption{The concurrent composition $Cc(\hat{G},Obs(G_{dss}))$ for the system $G$ in Fig.~\ref{Fig9}.}
  \label{Fig11}
\end{figure}
\end{example}

\section{Concluding remarks}\label{sec5}

In this paper, we characterized two strong versions of opacity called strong current-state opacity and strong initial-state opacity, respectively.
They have higher-level confidentiality than the standard versions.
Further, we constructed two information structures using a novel concurrent-composition technique to verify these two strong versions of opacity, which have (worst-case) time complexity $O(|X|^4|\Sigma_o||\Sigma_{uo}||\Sigma|2^{|X|})$.

Motivated by~\cite{Falcone(2015)} and our work, a strong version of standard Inf-SO, called \emph{strong infinite-step opacity} (Inf-SSO), is formally formulate in the context of nondeterministic settings as follows.
\begin{definition}[Inf-SSO]\label{de:5.1}
Given a system $G=(X,\Sigma,$ $\delta,X_0)$, a projection map $P$ w.r.t. the set $\Sigma_o$ of observable events, and a set $X_{S}\subset X$ of secret states,
$G$ is said to be strongly infinite-step opaque (Inf-SSO) (w.r.t. $\Sigma_o$ and $X_{S}$), if for all $x_0\in X_0$ and for all $s=s_1s_2\in\mathcal{L}(G,x_0)$ such that $\delta(x_0,s_1)\cap X_S\neq\emptyset$,
there exists a non-secret run $x^\prime_0\stackrel{t}{\rightarrow}$ such that $P(t)=P(s)$, where $x^\prime_0\in X_{NS}$.
\end{definition}

By Definition~\ref{de:5.1}, one sees that $G$ is Inf-SSO if and only if for all $x_0\in X_0$ and for all $s\in\mathcal{L}(G,x_0)$, there exists a non-secret run $x^\prime_0\stackrel{t}{\rightarrow}$ for some $x^\prime_0\in X_{NS}$ and some $t\in\mathcal{L}(G,x^\prime_0)$ such that $P(t)=P(s)$.
We can use the concurrent composition $Cc(G,Obs(G_{dss}))$ of $G$ and $Obs(G_{dss})$ constructed in Definition~\ref{de:3.2} to present the result on verification of Inf-SSO, which also has time complexity $O(|X|^4|\Sigma_o||\Sigma_{uo}||\Sigma|2^{|X|})$.
We here omit its proof, which is similar to that of Theorem~\ref{th:3.1}.

\begin{theorem}\label{th:5.1}
Given a system $G=(X,\Sigma,\delta,X_0)$, a projection map $P$ w.r.t. the set $\Sigma_o$ of observable events, and a set $X_{S}\subset X$ of secret states,
let $Cc(G,Obs(G_{dss}))$ be the concurrent composition of $G$ and $Obs(G_{dss})$.
$G$ is Inf-SSO w.r.t. $\Sigma_o$ and $X_S$ if and only if there exists no such a state in $Cc(G,Obs(G_{dss}))$ whose right component is $\emptyset$.
\end{theorem}

In the future, we plan to exploit the proposed concurrent-composition approach to design algorithms for enforcing strong current-state opacity, strong initial-state opacity and strong infinite-step opacity.

%\begin{ack}                               % Place acknowledgements
%This work is supported in part by the National Natural Science Foundation of China under Grant 61903274, and the Tianjin Natural Science Foundation of China under Grant 18JCQNJC74000.  % here.
%\end{ack}

\bibliographystyle{plain}        % Include this if you use bibtex
%\bibliography{autosam}           % and a bib file to produce the
                                 % bibliography (preferred). The
                                 % correct style is generated by
                                 % Elsevier at the time of printing.

\end{document}